\documentclass[11pt,letterpaper]{article}
\usepackage{fullpage}
\usepackage{times}

\newcommand{\qed}{\hfill\rule{2mm}{2mm}}
\newenvironment{proof}{\begin{trivlist}\item[\hspace{\labelsep}{\bf\noindent Proof. }]}{\end{trivlist}}

\newtheorem{theorem}{Theorem}
\newtheorem{lemma}[theorem]{Lemma}

\newtheorem{defn}[theorem]{Definition}

\title{\bf Efficient coordination mechanisms\\for unrelated machine scheduling\thanks{A preliminary version of the results of this paper appeared in {\em Proceedings of the 20th Annual ACM-SIAM Symposium on Discrete Algorithms (SODA)}, pp. 815-824, 2009.}}

\author{Ioannis Caragiannis\thanks{Research Academic Computer Technology Institute \& Department
of Computer Engineering and Informatics, University of Patras, 26500 Rio, Greece. Email: {\tt caragian@ceid.upatras.gr}.}}

\begin{document}

\maketitle

\begin{abstract}
We present three new coordination mechanisms for scheduling $n$ selfish jobs on $m$ unrelated machines.
A coordination mechanism aims to mitigate the impact of selfishness of jobs on the efficiency of schedules
by defining a local scheduling policy on each machine. The scheduling policies induce a game among
the jobs and each job prefers to be scheduled on a machine so that its completion time is minimum
given the assignments of the other jobs. We consider the maximum completion time among all jobs as the measure of
the efficiency of schedules. The approximation ratio of a coordination mechanism quantifies the
efficiency of pure Nash equilibria (price of anarchy) of the induced game.

Our mechanisms are deterministic, local, and preemptive in the sense that the scheduling policy does not necessarily
process the jobs in an uninterrupted way and may introduce some idle time.
Our first coordination mechanism has approximation ratio $\Theta(\log m)$ and always guarantees that the induced game has pure Nash equilibria
to which the system converges in at most $n$ rounds. This result improves a bound of $O(\log^2 m)$ due to Azar, Jain,
and Mirrokni and, similarly to their mechanism, our mechanism uses a global ordering of
the jobs according to their distinct IDs. Next we study the intriguing scenario where jobs are
anonymous, i.e., they have no IDs. In this case, coordination mechanisms can only distinguish between jobs
that have different load characteristics.
Our second mechanism handles anonymous jobs and has approximation ratio $O\left(\frac{\log m}{\log \log m}\right)$
although the game induced is not a potential game and, hence, the existence of pure Nash
equilibria is not guaranteed by potential function arguments. However, it provides evidence that the known lower
bounds for non-preemptive coordination mechanisms could be beaten using preemptive scheduling policies. Our third coordination mechanism also handles
anonymous jobs and has a nice ``cost-revealing'' potential function. 
We use this potential function in order, not only to prove the existence of equilibria, but also to upper-bound the price of stability of the induced game by $O(\log m)$ and
the price of anarchy by $O(\log^2m)$. Our third coordination
mechanism is the first that handles anonymous jobs and simultaneously guarantees that the induced game is a potential
game and has bounded price of anarchy. In order to obtain the above bounds, our coordination mechanisms use $m$ as a parameter. Slight variations of these mechanisms in which this information is not necessary achieve approximation ratios of $O\left(m^{\epsilon}\right)$, for any constant $\epsilon>0$.
\end{abstract}

\newpage
\section{Introduction}
We study the classical problem of {\em unrelated machine scheduling}.
In this problem, we have $m$ parallel machines and $n$ independent jobs. Job $i$ induces a (possibly infinite)
positive processing time (or load) $w_{ij}$ when processed by machine $j$. The load of a machine is the total
load of the jobs assigned to it. The quality of an assignment of jobs to machines is measured by the makespan (i.e., the maximum) of the machine loads
or, alternatively, the maximum completion time among all jobs. The optimization problem of computing an assignment
of minimum makespan is a fundamental APX-hard problem, quite well-understood in terms of its offline \cite{LST93}
and online approximability \cite{AAFPW93,ANR95}.

The approach we follow in this paper is both algorithmic and game-theoretic. We assume that each job is owned by a selfish agent.
This gives rise to a {\em selfish scheduling} setting where each agent aims to minimize the completion time of her job
with no regard to the global optimum. Such a selfish behaviour can lead to inefficient schedules from which no agent
has an incentive to unilaterally deviate in order to improve the completion time of her job. From the algorithmic point of view, the designer of such a system can define a {\em coordination mechanism} \cite{CKN04}, i.e., a
{\em scheduling policy} within each machine in order to ``coordinate'' the selfish behaviour of the jobs.
Our main objective is to design coordination mechanisms that guarantee that the assignments reached by the selfish
agents are {\em efficient}.

\medskip

\noindent {\bf The model.}
A scheduling policy simply defines the way jobs are scheduled within a machine and can be either {\em non-preemptive} or
{\em preemptive}. Non-preemptive scheduling policies process jobs uninterruptedly according to some order. Preemptive
scheduling policies do not necessarily have this feature and can also introduce some idle time (delay). Although this seems
unnecessary  at first glance, as we show in this paper, it is a very useful tool in order to guarantee coordination.
A coordination mechanism is a set of scheduling policies running on the machines. In the sequel, we use the terms coordination
mechanisms and scheduling policies interchangeably.

A coordination mechanism defines (or induces) a game with the job owners as players. Each job has all machines as possible {\em strategies}.
We call an {\em assignment} (of jobs to machines) or {\em state} any set of strategies selected by the players, with one strategy
per player. Given
an assignment of jobs to machines, the cost of a player is the completion time of her job on the machine it has
been assigned to; this completion time depends on the scheduling policy on that machine and the characteristics of all jobs
assigned to that machine. Assignments in which no player has an incentive to change her strategy in order to decrease
her cost given the assignments of the other players are called {\em pure Nash equilibria}. The global objective that is used in order to assess the efficiency of
assignments is the {\em maximum completion time} over all jobs. A related quantity is the {\em makespan} (i.e., the maximum
of the machine loads). Notice that when preemptive scheduling policies are used, these two quantities may not be
the same (since idle time contributes to the completion time but not to the load of a machine). However, the optimal makespan is a lower bound on the optimal maximum completion time. The {\em price of anarchy} \cite{P01} is
the maximum over all pure Nash equilibria of the ratio of the maximum completion time among all jobs over the optimal makespan.
The {\em price of stability} \cite{ADKTWR04} is the minimum over all pure Nash equilibria of the ratio of the maximum completion time among
all jobs over the optimal makespan. The {\em approximation ratio} of a coordination mechanism is the maximum of the price
of anarchy of the induced game over all input instances.

Four natural coordination mechanisms are the {\sf Makespan}, {\sf Randomized}, {\sf LongestFirst}, and {\sf ShortestFirst}.
In the {\sf Makespan} policy, each machine processes the jobs assigned to it ``in parallel'' so that the completion time of each
job is the total load of the machine. {\sf Makespan} is obviously a preemptive coordination mechanism.
In the {\sf Randomized} policy, the jobs are scheduled non-preemptively in random order. Here, the cost of each player
is the expected completion time of her job. In the {\sf ShortestFirst} and {\sf LongestFirst} policies, the jobs assigned
to a machine are scheduled in non-decreasing and non-increasing order of their processing times, respectively. In case of
ties, a {\em global ordering} of the jobs according to their distinct IDs is used. This is necessary by any deterministic
non-preemptive coordination mechanism in order to be well-defined. Note that no such information is required
by the {\sf Makespan} and {\sf Randomized} policies; in this case, we say that they handle {\em anonymous jobs}. According to
the terminology of \cite{AJM08}, all these
four coordination mechanisms are {\em strongly local} in the sense that the only information required by each machine
in order to compute a schedule are the processing times of the jobs assigned to it. A {\em local} coordination mechanism
may use all parameters (i.e., the load vector) of the jobs assigned to the same machine.

Designing coordination mechanisms with as small approximation ratio as possible is our main concern. But there are other
issues related to efficiency. The price of anarchy is meaningful only in games where pure Nash equilibria exist. So, the primary goal
of the designer of a coordination mechanism should be that the induced game {\em always has}
pure Nash equilibria. Furthermore, these equilibria should be {\em easy to find}. A very interesting class of games in which
the existence of pure Nash equilibria is guaranteed is that of {\em potential} games. These games have the property that
a {\em potential function} can be defined on the states of the game so that in any two states differing in the strategy of a
single player, the difference of the values of the potential function
and the difference of the cost of the player have the same sign. This property guarantees that the state with minimum potential
is a pure Nash equilibrium. Furthermore, it guarantees that, starting from any state, the system will reach (converge to)
a pure Nash equilibrium after a finite number of {\em selfish moves}. Given a game, its {\em Nash dynamics}
is a directed graph with the states of the game as nodes and edges connecting two states differing in the strategy of a single player if
that player has an incentive to change her strategy according to the direction of the edge. The Nash dynamics of potential
games do not contain any cycle. Another desirable property here is {\em fast} convergence,
i.e., convergence to a pure Nash equilibrium in a polynomial number of selfish moves. A particular type of selfish moves
that have been extensively considered in the literature \cite{AAEMS08,CMS06,FFM08,MV04} is that of {\em best-response} moves.
In a best-response move, a player having an incentive to change her strategy selects the strategy that yields the maximum
decrease in her cost.

Potential games are strongly related to {\em congestion games} introduced by Rosenthal \cite{R73}. Rosenthal presented a potential
function for these games with the following property: in any two states differing in the strategy of a single player,
the difference of the values of the potential function {\em equals} the difference of the cost of the player. Monderer and Shapley
\cite{MS96} have proved that each potential game having this property is isomorphic to a congestion game. We point out that potential
functions are not the only way to guarantee the existence of pure Nash equilibria. Several generalizations of congestion games such as those with
player-specific latency functions \cite{M96} are not potential games but several subclasses of them provably have pure Nash equilibria.

\medskip

\noindent {\bf Related work.} The study of the price of anarchy of games began with the seminal work of Koutsoupias and Papadimitriou \cite{KP99}
and has played a central role in the recently emerging field of Algorithmic Game Theory \cite{NRTV07}. Several papers provide bounds on the price of anarchy
of different games of interest. Our work follows a different direction where the price of anarchy is the {\em objective
to be minimized} and, in this sense, it is similar in spirit to studies where the main question is how to change the
rules of the game at hand in order to improve the price of anarchy. Typical examples are the introduction of taxes or tolls in congestion games \cite{CKK06,CDR06,FJM04,KK04,S07},
protocol design in network and cost allocation games \cite{CRV08,KLO95}, Stackelberg routing strategies \cite{KS06,KLO97,KM02,R04,S07}, and network design \cite{R06}.

Coordination mechanisms were introduced by Christodoulou, Koutsoupias, and Nanavati in \cite{CKN04}. They study the case where each player
has the same load on each machine and, among other results, they consider the {\sf LongestFirst} and {\sf ShortestFirst} scheduling policies.
We note that the {\sf Makespan} and {\sf Randomized} scheduling policies were used in \cite{KP99} as models of selfish behaviour in scheduling,
and since that paper, the {\sf Makespan} policy has been considered as standard in the study of selfish scheduling games in models simpler than
the one of unrelated machines and is strongly related to
the study of congestion games (see \cite{V07,R07} and the references therein). Immorlica et al. \cite{ILMS05} study these four scheduling policies
under several scheduling settings including
the most general case of unrelated machines. They prove that the {\sf Randomized} and {\sf ShortestFirst} policies have
approximation ratio $O(m)$ while the {\sf LongestFirst} and {\sf Makespan} policies have unbounded approximation ratio.
Some scheduling policies are also related to earlier studies of local-search scheduling heuristics. So, the fact that the
price of anarchy of the induced game may be unbounded follows by the work of Schuurman and Vredeveld \cite{SV01}.
As observed in \cite{ILMS05}, the equilibria of the game induced by {\sf ShortestFirst} correspond to the solutions of the {\sf ShortestFirst}
scheduling heuristic which is known to be $m$-approximate \cite{IK77}.
The {\sf Makespan} policy is known to induce potential games \cite{EKM03}. The {\sf ShortestFirst} policy also induces potential
games as proved in \cite{ILMS05}. In Section \ref{sec:b-cm}, we present examples showing that the scheduling policies {\sf LongestFirst} and {\sf Randomized} do not
induce potential games.\footnote{After the appearance of the conference version of the paper, we became aware of two independent proofs that {\sf Longest-First} may induce games that do not have pure Nash equilibria \cite{DT09,FRS10}.}

Azar et al. \cite{AJM08} study non-preemptive coordination mechanisms for unrelated machine scheduling. They prove that
any local non-preemptive coordination mechanism is at least $\Omega(\log m)$-approximate\footnote{The corresponding proof of \cite{AJM08} contained a error which has been recently fixed by Fleischer and Svitkina \cite{FS10}.} while any strongly local
non-preemptive coordination mechanism is at least $\Omega(m)$-approximate; as a corollary, they solve an old open
problem concerning the approximation ratio of the {\sf ShortestFirst} heuristic. 
On the positive side, the authors of \cite{AJM08} present a non-preemptive
local coordination mechanism (henceforth called {\sf AJM-1}) that is $O(\log m)$-approximate although it may induce games without pure Nash equilibria.
The extra information used by this scheduling policy is the {\em inefficiency} of jobs (defined in the next section).
They also present a technique that transforms this coordination mechanism to a preemptive one that induces potential games
with price of anarchy $O(\log^2 m)$. In their mechanism, the players converge to a pure Nash equilibrium in $n$ rounds of best-response
moves. We will refer to this coordination mechanism as {\sf AJM-2}. Both {\sf AJM-1} and {\sf AJM-2} use the IDs of the jobs.


\medskip

\noindent {\bf Our results.} We present three new coordination mechanisms for unrelated machine scheduling.
Our mechanisms are deterministic, preemptive, and local. The schedules in each machine are computed as functions of the characteristics
of jobs assigned to the machine, namely the load of jobs on the machine and their inefficiency. In all cases, the
functions use an integer parameter $p\geq 1$; the best
choice of this parameter for our coordination mechanisms is $p=O(\log m)$. Our analysis is heavily based on the convexity of simple polynomials and geometric inequalities for Euclidean norms.

\begin{table}
\centerline{\footnotesize
\begin{tabular}{|l||c|c|c|c|l|}
  \hline
  Coordination           &  &   &   &                &   \\
  mechanism           & PoA                                       & Pot. & PNE & IDs  & Characteristics \\\hline\hline
  {\sf ShortestFirst} & $\Theta(m)$                               & Yes  & Yes & Yes  & Strongly local, non-preemptive\\\hline
  {\sf LongestFirst}  & unbounded                                 & No   & No   & Yes  & Strongly local, non-preemptive\\\hline
  {\sf Makespan}      & unbounded                                 & Yes  & Yes & No   & Strongly local, preemptive\\\hline
  {\sf Randomized}    & $\Theta(m)$                               & No   & ?   & No   & Strongly local, non-preemptive\\\hline
  {\sf AJM-1}         & $\Theta(\log m)$                          & No   & No  & Yes  & Local, non-preemptive\\\hline
  {\sf AJM-2}         & $O(\log^2m)$                              & Yes  & Yes & Yes  & Local, preemptive, uses $m$ \\\hline\hline
  {\sf ACOORD}        & $\Theta(\log m)$                          & Yes  & Yes & Yes  & Local, preemptive, uses $m$ \\
                      & $O(m^\epsilon)$                           & Yes  & Yes & Yes  & Local, preemptive \\\hline
  {\sf BCOORD}        & $O\left(\frac{\log m}{\log\log m}\right)$ & No   & ?   & No   & Local, preemptive, uses $m$ \\
                      & $O(m^\epsilon)$                           & No   & ?   & No   & Local, preemptive \\\hline
  {\sf CCOORD}        & $O(\log^2m)$                              & Yes  & Yes & No   & Local, preemptive, uses $m$ \\
                      & $O(m^\epsilon)$                           & Yes  & Yes & No   & Local, preemptive \\\hline
\end{tabular}}
\label{tab:comparison}
\caption{Comparison of our coordination mechanisms to previously known ones with respect to the price of anarchy of the
induced game (PoA), whether they induced potential games or not (Pot.), the existence of pure Nash equilibria (PNE),
and whether they use the job IDs or not.}
\end{table}

Motivated by previous work, we first consider the scenario where jobs have distinct IDs.
Our first coordination mechanism {\sf ACOORD} uses this information and is superior to the known coordination mechanisms
that induce games with pure Nash equilibria. The game induced
is a potential game, has price of anarchy $\Theta(\log m)$, and the players converge to pure Nash equilibria in at most $n$
rounds. Essentially, the equilibria of the game induced by {\sf ACOORD}
can be thought of as the solutions produced by the application of a particular online algorithm, similar to the greedy online
algorithm for minimizing the $\ell_p$ norm
of the machine loads \cite{AAFPW93,C08}. Interestingly, the local objective of the greedy online algorithm for the $\ell_p$ norm
may not translate to a completion time of jobs in feasible schedules; the online algorithm implicit by {\sf ACOORD} uses
a different local objective that meets this constraint. The related results are presented in Section \ref{sec:a-cm}.

Next we address the case where no ID information is associated to the jobs (anonymous jobs). This scenario is relevant
when the job owners do not wish to reveal their identities or in large-scale settings where distributing IDs to jobs is infeasible.
Definitely, an advantage that could be used for coordination is lost in this way but this makes the problem of designing coordination mechanisms
more challenging. In Section \ref{sec:b-cm}, we present our second coordination mechanism {\sf BCOORD} which induces a simple congestion game with player-specific polynomial latency
functions of a particular form. The price of anarchy of this game is only $O\left(\frac{\log m}{\log \log m}\right)$.
This result demonstrates that preemption may be useful in order to beat the $\Omega(\log m)$ lower bound of \cite{AJM08}
for non-preemptive coordination mechanisms. On the negative side, we show that the game induced may not be a potential
game by presenting an example where the Nash dynamics have a cycle.

Our third coordination mechanism {\sf CCOORD} is presented in Section \ref{sec:c-cm}. The scheduling policy on each machine
uses an interesting function on the loads of the jobs assigned to the machine and their inefficiency. The game induced by {\sf CCOORD} is a potential
game; the associated potential function is ``cost-revealing'' in the sense that it can be used to upper-bound the cost of equilibria.
In particular, we show that the price of stability
of the induced game is $O(\log m)$ and the price of anarchy is $O(\log^2 m)$. The coordination
mechanism {\sf CCOORD} is the first that handles anonymous jobs and simultaneously guarantees that the induced game is a potential
game and has bounded price of anarchy. Table \ref{tab:comparison} compares our coordination mechanisms to the previously known ones.

Observe that the dependence of the parameter $p$ on $m$ requires that our mechanisms use the number of machines as input. By setting $p$ equal to an appropriately large constant, our mechanisms achieve price of anarchy  $O(m^\epsilon)$ for any constant $\epsilon>0$. In particular, the coordination mechanisms {\sf ACOORD} and {\sf CCOORD} are the first ones that do not use the number of machines as a parameter, induce games with pure Nash equilibria, and have price of anarchy $o(m)$.

We remark that the current paper contains several improvements compared to its conference version. There, the three coordination mechanisms had the restriction that a job with inefficiency more than $m$ on some machine has infinite completion time when assigned to that machine. Here, we have removed this restriction and have adapted the analysis accordingly. A nice consequence of the new definition is that the coordination mechanisms can now be defined so that they do not use the number of machines as a parameter.
Furthermore, the definition of {\sf ACOORD} has been significantly simplified. Also, the analysis of the price of anarchy of the coordination mechanism {\sf BCOORD} in the conference version used a technical lemma which is implicit in \cite{STZ04}. In the current version, we present a different self-contained proof that is based on convexity properties of polynomials and Minkowski inequality; the new proof has a similar structure with the analysis of the price of anarchy of mechanism {\sf ACOORD}.

We begin with preliminary technical definitions in Section \ref{sec:prelim} and conclude with interesting open questions in Section
\ref{sec:open}.

\section{Preliminaries}\label{sec:prelim}
In this section, we present our notation and give some statements that will be useful later.
We reserve $n$ and $m$ for the number of jobs and machines, respectively, and the indices $i$ and $j$
for jobs and machines, respectively. Unless specified otherwise, the sums $\sum_i$ and $\sum_j$ run over
all jobs and over all machines, respectively. Assignments are denoted by $N$ or $O$. With some abuse in
notation, we use $N_j$ to denote both the set of jobs assigned to machine $j$ and the set of their loads
on machine $j$. We use the notation $L(N_j)$ to denote the load of machine $j$ under the assignment $N$.
More generally, $L(A)$ denotes the sum of the elements for any set of non-negative reals $A$.
For an assignment $N$ which assigns job $i$ to machine $j$, we denote the completion time
of job $i$ under a given scheduling policy by ${\cal P}(i,N_j)$.
Note that, besides defining the completion times, we do not discuss the particular way the jobs are scheduled by
the scheduling policies we present. However, we require that {\em feasible} schedules are computable efficiently. A natural
sufficient and necessary condition is the following: for any job $i\in N_j$, the total load of jobs with
completion time at most ${\cal P}(i,N_j)$ is at most ${\cal P}(i,N_j)$.

Our three coordination mechanisms use the inefficiency of jobs in order to compute schedules. We
denote by $w_{i,\min}$ the minimum load of job $i$ over all machines. Then, its inefficiency $\rho_{ij}$
on machine $j$ is defined as $\rho_{ij}=w_{ij}/w_{i,\min}$.

Our proofs are heavily based on the convexity of simple polynomials such as $z^k$ for $k\geq 1$ and on the relation
of Euclidean norms of the machine loads and the makespan. Recall that the $\ell_k$ norm of the machine loads for an
assignment $N$ is $\left(\sum_j{L(N_j)^k}\right)^{1/k}$. The proof of the next lemma is trivial.
\begin{lemma}\label{lem:lp-norm}
For any assignment $N$,
$\max_j{L(N_j)} \leq \left(\sum_j{L(N_j)^k}\right)^{1/k} \leq m^{1/k} \max_j{L(N_j)}$.
\end{lemma}
In some of the proofs, we also use the Minkowski inequality (or the triangle inequality for the $\ell_p$ norm).

\begin{lemma}[Minkowski inequality]\label{lem:minkowski}
$\left(\sum_{t=1}^s{(a_t+b_t)^k}\right)^{1/k} \leq \left(\sum_{t=1}^s{a_t^k}\right)^{1/k}+\left(\sum_{t=1}^s{b_t^k}\right)^{1/k}$, for any $k\geq 1$ and $a_t,b_t\geq 0$.
\end{lemma}

The following two technical lemmas are used in some of our proofs. We include them here for easy reference.

\begin{lemma}\label{lem:convexity}
Let $r\geq 1$, $t\geq 0$ and $a_i\geq 0$, for $i=1, ..., k$. Then,
\[\sum_{i=1}^k{\left(\left(t+a_i\right)^r-t^r\right)}\leq
\left(t+\sum_{i=1}^k{a_i}\right)^r-t^r\]
\end{lemma}

\begin{proof}
The case when $a_i=0$ for $i=1, ..., k$ is trivial. Assume otherwise
and let $\xi=\sum_{i=1}^{k}{a_i}$ and $\xi_i=a_i/\xi$. Clearly,
$\sum_{i=1}^k{\xi_i}=1$. By the convexity of function $z^r$ in
$[0,\infty)$, we have that
\begin{eqnarray}\nonumber
(t+a_i)^r &= &
\left((1-\xi_i)t+\xi_i\left(t+\sum_{i=1}^k{a_i}\right)\right)^r\\\label{eq:convexity}
&\leq &
(1-\xi_i)t^r+\xi_i\left(t+\sum_{i=1}^k{a_i}\right)^r
\end{eqnarray}
for $i=1, ..., k$. Using (\ref{eq:convexity}), we obtain
\begin{eqnarray*}
\sum_{i=1}^k{\left((t+a_i)^r-t^r\right)} &\leq &
t^r\left(\sum_{i=1}^k{(1-\xi_i)}-k\right)+\left(t+\sum_{i=1}^k{a_i}\right)^r\sum_{i=1}^k{\xi_i}\\
&=& \left(t+\sum_{i=1}^k{a_i}\right)^r-t^r
\end{eqnarray*}
\qed\end{proof}

\begin{lemma}\label{lem:slope}
For any $z_0\geq 0$, $\alpha\geq 0$, and $p\geq 1$, it holds \[(p+1)\alpha z_0^p \leq (z_0+\alpha)^{p+1}-z_0^{p+1}\leq (p+1)\alpha (z_0+\alpha)^p.\]
\end{lemma}
\begin{proof}
The inequality trivially holds if $\alpha=0$. If $\alpha>0$, the inequality follows since, due to the convexity of the function $z^{p+1}$, the slope of the line that crosses points $(z_0,z_0^{p+1})$ and $(z_0+\alpha,(z_0+\alpha)^{p+1})$ is between its derivative at points $z_0$ and $z_0+\alpha$.\qed\end{proof}

We also refer to the multinomial and binomial theorems. \cite{HLP52} provides an extensive overview of the
inequalities we use and their history (see also {\tt wikipedia.org} for a quick survey).

\section{The coordination mechanism {\sf ACOORD}}\label{sec:a-cm}
The coordination mechanism {\sf ACOORD} uses a global ordering of the jobs according to their distinct IDs.
Without loss of generality, we may assume that the index of a job is its ID.
Let $N$ be an assignment and denote by $N^i$ the restriction of $N$ to the jobs with the $i$ smallest IDs.
{\sf ACOORD} schedules job $i$ on machine $j$ so that it completes at time
\[{\cal P}(i,N_j) = \left(\rho_{ij}\right)^{1/p} L(N^i_j).\]
Since $\rho_{ij}\geq 1$, the schedules produced are always feasible.

Consider the sequence of jobs in increasing order of their IDs and assume that each job plays a best-response
move. In this case, job $i$ will select that machine $j$ so that the quantity $\left(\rho_{ij}\right)^{1/p} L(N^i_j)$
is minimized. Since the completion time of job $i$ depends only on jobs with smaller IDs,
no job will have an incentive to change its strategy and the resulting assignment is a pure Nash
equilibrium. The following lemma extends this observation in a straightforward way.

\begin{lemma}\label{lem:a-cm-potential-convergence}
The game induced by the coordination mechanism {\sf ACOORD} is a potential game.
Furthermore, any sequence of $n$ rounds of best-response moves
converges to a pure Nash equilibrium.
\end{lemma}

\begin{proof}
Notice that since a job does not affect the completion time of jobs with smaller IDs, the vector of completion times of the jobs (sorted in increasing order of their IDs) decreases
lexicographically when a job improves its cost by deviating to another strategy and, hence, it is a potential
function for the game induced by the coordination mechanism {\sf ACOORD}.

Now, consider $n$ rounds of best-response moves of the jobs in the induced game such that each job plays at least
once in each round. It is not hard to see that after round $i$, the job $i$ will have selected that machine $j$
so that the quantity $\left(\rho_{ij}\right)^{1/p} L(N^i_j)$ is minimized. Since the completion time of job
$i$ depends only on jobs with smaller IDs, job $i$ has no incentive to move after round $i$ and, hence, no job
will have an incentive to change its strategy after the $n$ rounds. So, the resulting assignment is a pure Nash equilibrium.
\qed
\end{proof}

The sequence of best-response moves mentioned above can be thought of as an online algorithm that processes the jobs in increasing
order of their IDs. The local objective is slightly different that the local objective of the greedy online
algorithm for minimizing the $\ell_{p+1}$ norm of the machine loads \cite{AAG+95,C08}; in that algorithm, job $i$ is assigned
to a machine $j$ so that the quantity $(L(N^{i-1}_j)+w_{ij})^{p+1}-L(N^{i-1}_j)^{p+1}$ is minimized. Here, we remark that we do not see how the
local objective of that algorithm could be simulated by a scheduling policy that always produces feasible schedules.
This constraint is trivially satisfied by the coordination mechanism {\sf ACOORD}. The next lemma bounds the maximum completion time at pure Nash equilibria in terms of the $\ell_{p+1}$ norm of the machine loads and the optimal makespan.

\begin{lemma}\label{lem:completion-acoord}
Let $N$ be a pure Nash equilibrium of the game induced by the coordination mechanisms {\sf ACOORD} and let $O$ be an optimal assignment. Then
\[\max_{j,i\in N_j}{{\cal P}(i,N_j)} \leq \left(\sum_j{L(N_j)^{p+1}}\right)^{\frac{1}{p+1}} +\max_j{L(O_j)}.\]
\end{lemma}

\begin{proof}
Let $i^*$ be the job that has the maximum completion time in assignment $N$. Denote by $j_1$ the machine $i^*$ uses in $N$ and let $j_2$ be a machine such that $\rho_{i^*j_2}=1$. If $j_1=j_2$, the definition of the coordination mechanism {\sf ACOORD} yields
\begin{eqnarray*}
\max_{j,i\in N_j}{{\cal P}(i,N_j)} &=& {\cal P}(i^*,N_{j_1})\\
&=& L(N_{j_1}^{i^*})\\
&\leq & L(N_{j_1})\\
&\leq & \left(\sum_j{L(N_j)^{p+1}}\right)^{\frac{1}{p+1}}.
\end{eqnarray*}
Otherwise, since player $i^*$ has no incentive to use machine $j_2$ instead of $j_1$, we have
\begin{eqnarray*}
\max_{j,i\in N_j}{{\cal P}(i,N_j)} &=& {\cal P}(i^*,N_{j_1})\\
&\leq & {\cal P}(i^*,N_{j_2}\cup \{w_{i^*j_2}\})\\
&=& L(N_{j_2}^{i^*})+w_{i^*j_2}\\
&\leq & L(N_{j_2})+\min_j{w_{i^*j}}\\
&\leq & \left(\sum_j{L(N_j)^{p+1}}\right)^{\frac{1}{p+1}} +\max_j{L(O_j)}.
\end{eqnarray*}
\qed
\end{proof}

Next we show that the approximation ratio of {\sf ACOORD} is $O(\log m)$ (for well-selected values of the
parameter $p$). The analysis borrows and extends techniques
from the analysis of the greedy online algorithm for the $\ell_p$ norm in \cite{C08}.

\begin{theorem}\label{thm:a-cm-poa}
The price of anarchy of the game induced by the coordination mechanism {\sf ACOORD} with $p=\Theta(\log m)$ is $O(\log m)$. Also, for every constant $\epsilon\in (0,1/2]$, the price of anarchy of the game induced by the coordination mechanism {\sf ACOORD} with $p=1/\epsilon-1$ is $O\left(m^\epsilon\right)$.
\end{theorem}

\begin{proof}
Consider a pure Nash equilibrium $N$ and an optimal assignment $O$. Since no job has an incentive to change her strategy from $N$, for any job $i$ that is assigned to machine $j_1$ in $N$ and to machine $j_2$ in $O$, by the definition of {\sf ACOORD} we have that
\begin{eqnarray*}\left(\rho_{ij_1}\right)^{1/p} L(N_{j_1}^i) &\leq & \left(\rho_{ij_2}\right)^{1/p} \left(L(N_{j_2}^{i-1})+w_{ij_2}\right).
\end{eqnarray*}
Equivalently, by raising both sides to the power $p$ and multiplying with $w_{i,\min}$, we have that
\begin{eqnarray*}w_{ij_1} L(N_{j_1}^i)^p &\leq & w_{ij_2} \left(L(N_{j_2}^{i-1})+w_{ij_2}\right)^p.\end{eqnarray*}

Using the binary variables $x_{ij}$ and $y_{ij}$ to denote whether job $i$ is assigned to machine $j$ in the assignment $N$ ($x_{ij}=1$) and $O$ ($y_{ij}=1$), respectively,
or not ($x_{ij}=0$ and $y_{ij}=0$, respectively), we can express this last inequality as follows.
\begin{eqnarray*}
\sum_j{x_{ij}w_{ij}L(N_{j}^i)^p} &\leq &\sum_j{y_{ij}w_{ij}\left(L(N_{j}^{i-1})+w_{ij}\right)^{p}}
\end{eqnarray*}
By summing over all jobs and multiplying with $(e-1)(p+1)$, we have
\begin{eqnarray}\nonumber
& & (e-1)(p+1) \sum_i\sum_j{x_{ij}w_{ij}L(N_j^i)^p}\\\nonumber
&\leq & (e-1)(p+1) \sum_i\sum_j{y_{ij}w_{ij}\left(L(N_j^{i-1})+w_{ij}\right)^p}\\\nonumber
&\leq & (e-1)(p+1) \sum_j\sum_i{y_{ij}w_{ij}\left(L(N_j)+w_{ij}\right)^p}\\\nonumber
&= & (e-1)(p+1) \sum_j\sum_i{y_{ij}w_{ij}\left(L(N_j)+y_{ij}w_{ij}\right)^p}\\\nonumber
&\leq & \sum_j\sum_i{\left(\left(L(N_j)+ey_{ij}w_{ij}\right)^{p+1}-\left(L(N_j)+y_{ij}w_{ij}\right)^{p+1}\right)}\\\nonumber
&\leq & \sum_j\sum_i{\left(\left(L(N_j)+ey_{ij}w_{ij}\right)^{p+1}-L(N_j)^{p+1}\right)}\\\nonumber
&\leq & \sum_j{\left(\left(L(N_j)+e\sum_i{y_{ij}w_{ij}}\right)^{p+1}-L(N_j)^{p+1}\right)}\\\nonumber
&=& \sum_j{\left(L(N_j)+eL(O_j)\right)^{p+1}}-\sum_j{L(N_j)^{p+1}}\\\label{eq:a-minkowski}
&\leq & \left(\left(\sum_j{L(N_j)^{p+1}}\right)^{\frac{1}{p+1}}+e\left(\sum_j{L(O_j)^{p+1}}\right)^{\frac{1}{p+1}}\right)^{p+1}-\sum_j{L(N_j)^{p+1}}.
\end{eqnarray}
The second inequality follows by exchanging the sums and since $L(N_j^{i-1})\leq L(N_j)$, the first equality follows since $y_{ij}\in \{0,1\}$, the third inequality follows by applying Lemma \ref{lem:slope} with $\alpha=(e-1)y_{ij}w_{ij}$ and $z_0=L(N_j)+y_{ij}w_{ij}$, the fourth inequality is obvious, the fifth inequality follows by Lemma \ref{lem:convexity}, the second equality follows since the definition of the variables $y_{ij}$ implies that $L(O_j)=\sum_i{y_{ij}w_{ij}}$, and the last inequality follows by Minkowski inequality (Lemma \ref{lem:minkowski}).

Now, we will relate the $\ell_{p+1}$ norm of the machines loads of assignments $N$ and $O$. We have
\begin{eqnarray*}
(e-1)\sum_j{L(N_j)^{p+1}} &=& (e-1)\sum_j{L(N_j^n)^{p+1}}\\
&=& (e-1)\sum_{i=1}^n{\sum_j{\left(L(N_j^i)^{p+1}-L(N_j^{i-1})^{p+1}\right)}}\\
&=& (e-1)\sum_{i=1}^n{\sum_j{\left(L(N_j^i)^{p+1}-(L(N_j^{i})-x_{ij}w_{ij})^{p+1}\right)}}\\
&\leq & (e-1)(p+1)\sum_i\sum_j{x_{ij}w_{ij}L(N_j^i)^p}\\
&\leq & \left(\left(\sum_j{L(N_j)^{p+1}}\right)^{\frac{1}{p+1}}+e\left(\sum_j{L(O_j)^{p+1}}\right)^{\frac{1}{p+1}}\right)^{p+1}-\sum_j{L(N_j)^{p+1}}.
\end{eqnarray*}
The first two equalities are obvious (observe that $L(N_j^0)=0$), the third one follows by the definition of variables $x_{ij}$, the first inequality follows by applying Lemma \ref{lem:slope} with $\alpha=x_{ij}w_{ij}$ and $z_0=L(N_j^{i})-x_{ij}w_{ij}$, and the last inequality follows by inequality (\ref{eq:a-minkowski}).

So, the above inequality yields
\begin{eqnarray*}
\left(\sum_j{L(N_j)^{p+1}}\right)^{\frac{1}{p+1}} &\leq & \frac{e}{e^{\frac{1}{p+1}}-1} \left(\sum_j{L(O_j)^{p+1}}\right)^{\frac{1}{p+1}}\\
&\leq & e(p+1) \left(\sum_j{L(O_j)^{p+1}}\right)^{\frac{1}{p+1}}\\
&\leq & e(p+1)m^{\frac{1}{p+1}} \max_j{L(O_j)}.
\end{eqnarray*}
The second inequality follows since $e^z\geq z+1$ for $z\geq 0$ and the third one follows by Lemma \ref{lem:lp-norm}.

Now, using Lemma \ref{lem:completion-acoord} and this last inequality, we obtain that
\begin{eqnarray*}
\max_{j,i\in N_j}{{\cal P}(i,N_j)} &\leq &\left(\sum_j{L(N_j)^{p+1}}\right)^{\frac{1}{p+1}} +\max_j{L(O_j)}\\
&\leq & \left(e(p+1)m^{\frac{1}{p+1}}+1\right) \max_j{L(O_j)}.
\end{eqnarray*}
The desired bounds follow by setting $p=\Theta(\log m)$ and $p=1/\epsilon-1$, respectively. \qed
\end{proof}
Our logarithmic bound is asymptotically tight; this follows by the connection to online algorithms mentioned above and the lower bound of \cite{ANR95}.

\section{The coordination mechanism {\sf BCOORD}}\label{sec:b-cm}
We now turn our attention to coordination mechanisms that handle anonymous jobs. We define the coordination mechanism {\sf BCOORD} by slightly changing the definition of {\sf ACOORD} so that the completion time of a job does not depend on its ID. So, {\sf BCOORD} schedules job $i$ on machine $j$ so that it finishes at time $${\cal P}(i,N_j)=\left(\rho_{ij}\right)^{1/p}L(N_j).$$ Since $\rho_{ij}\geq 1$, the schedules produced are always feasible.
The next lemma bounds the maximum completion time at pure Nash equilibria (again in terms of the $\ell_{p+1}$ norm of the machine loads and the optimal makespan).

\begin{lemma}\label{lem:completion-bcoord}
Let $N$ be a pure Nash equilibrium of the game induced by the coordination mechanisms {\sf BCOORD} and let $O$ be an optimal assignment. Then
\[\max_{j,i\in N_j}{{\cal P}(i,N_j)} \leq \left(\sum_j{L(N_j)^{p+1}}\right)^{\frac{1}{p+1}} +\max_j{L(O_j)}.\]
\end{lemma}

\begin{proof}
The proof is almost identical to the proof of Lemma \ref{lem:completion-acoord}; we include it here for completeness. Let $i^*$ be the job that has the maximum completion time in assignment $N$. Denote by $j_1$ the machine $i^*$ uses in $N$ and let $j_2$ be a machine such that $\rho_{i^*j_2}=1$. If $j_1=j_2$, the definition of the coordination mechanism {\sf BCOORD} yields
\begin{eqnarray*}
\max_{j,i\in N_j}{{\cal P}(i,N_j)} &=& {\cal P}(i^*,N_{j_1})\\
&=& L(N_{j_1})\\
&\leq & \left(\sum_j{L(N_j)^{p+1}}\right)^{\frac{1}{p+1}}.
\end{eqnarray*}
Otherwise, since player $i^*$ has no incentive to use machine $j_2$ instead of $j_1$, we have
\begin{eqnarray*}
\max_{j,i\in N_j}{{\cal P}(i,N_j)} &=& {\cal P}(i^*,N_{j_1})\\
&\leq & {\cal P}(i^*,N_{j_2}\cup \{w_{i^*j_2}\})\\
&=& L(N_{j_2})+w_{i^*j_2}\\
&= & L(N_{j_2})+\min_j{w_{i^*j}}\\
&\leq & \left(\sum_j{L(N_j)^{p+1}}\right)^{\frac{1}{p+1}} +\max_j{L(O_j)}.
\end{eqnarray*}\qed\end{proof}

We are ready to present our upper bounds on the price of anarchy of the induced game.

\begin{theorem}
The price of anarchy of the game induced by the coordination mechanism {\sf BCOORD} with $p=\Theta(\log m)$ is $O\left(\frac{\log m}{\log \log m}\right)$.
Also, for every constant $\epsilon\in (0,1/2]$, the price of anarchy of the game induced by the coordination mechanism {\sf BCOORD} with $p=1/\epsilon-1$ is $O\left(m^\epsilon\right)$.
\end{theorem}

\begin{proof}
Consider a pure Nash equilibrium $N$ and an optimal assignment $O$. Since no job has an incentive to change her strategy from $N$, for any job $i$ that
is assigned to machine $j_1$ in $N$ and to machine $j_2$ in $O$, we have that
\begin{eqnarray*}(\rho_{ij_1})^{1/p}L(N_{j_1}) &\leq & (\rho_{ij_2})^{1/p}(L(N_{j_2})+w_{ij_2}).
\end{eqnarray*}
Equivalently, by raising both sides to the power $p$ and multiplying both sides with $w_{i,\min}$, we have that
\begin{eqnarray*}w_{ij_1}L(N_{j_1})^p &\leq &w_{ij_2}(L(N_{j_2})+w_{ij_2})^p.\end{eqnarray*}
Using the binary variables $x_{ij}$ and $y_{ij}$ to denote whether job $i$ is assigned to machine $j$ in the assignments $N$ ($x_{ij}=1$) and $O$ ($y_{ij}=1$), respectively,
or not ($x_{ij}=0$ and $y_{ij}=0$, respectively), we can express this last inequality as follows:
\begin{eqnarray*}
\sum_j{x_{ij}w_{ij}L(N_j)^p} &\leq & \sum_j{y_{ij}w_{ij}(L(N_j)+w_{ij})^p}.
\end{eqnarray*}
By summing over all jobs and multiplying with $p$, we have
\begin{eqnarray}\nonumber
& & p\sum_i{\sum_j{x_{ij}w_{ij}L(N_j)^p}}\\\nonumber
&\leq & p\sum_i{\sum_j{y_{ij}w_{ij}(L(N_j)+w_{ij})^p}}\\\nonumber
&=& p\sum_j{\sum_i{y_{ij}w_{ij}(L(N_j)+y_{ij}w_{ij})^p}}\\\nonumber
&\leq & \sum_j\sum_i{\left(\left(L(N_j)+\frac{2p+1}{p+1}y_{ij}w_{ij}\right)^{p+1}-\left(L(N_j)+y_{ij}w_{ij}\right)^{p+1}\right)}\\\nonumber
&\leq & \sum_j\sum_i{\left(\left(L(N_j)+\frac{2p+1}{p+1}y_{ij}w_{ij}\right)^{p+1}-L(N_j)^{p+1}\right)}\\\nonumber
&\leq & \sum_j{\left(\left(L(N_j)+\frac{2p+1}{p+1}\sum_i{y_{ij}w_{ij}}\right)^{p+1}-L(N_j)^{p+1}\right)}\\\nonumber
&=& \sum_j{\left(\left(L(N_j)+\frac{2p+1}{p+1}L(O_j)\right)^{p+1}-L(N_j)^{p+1}\right)}\\\label{eq:b-minkowski}
&\leq& \left(\left(\sum_j{L(N_j)^{p+1}}\right)^{\frac{1}{p+1}}+\frac{2p+1}{p+1}\left(\sum_j{L(O_j)^{p+1}}\right)^{\frac{1}{p+1}}\right)^{p+1}-\sum_j{L(N_j)^{p+1}}
\end{eqnarray}
The first equality follows by exchanging the sums and since $y_{ij}\in \{0,1\}$, the second inequality follows by applying Lemma \ref{lem:slope} with $\alpha=\frac{p}{p+1}y_{ij}w_{ij}$ and $z_0=L(N_j)+y_{ij}w_{ij}$, the third inequality is obvious, the fourth inequality follows by applying Lemma \ref{lem:convexity}, the second equality follows since the definition of variables $y_{ij}$ implies that $L(O_j)=\sum_i{y_{ij}w_{ij}}$, and the last inequality follows by applying Minkowski inequality (Lemma \ref{lem:minkowski}).

Now, we relate the $\ell_{p+1}$ norm of the machine loads of assignments $N$ and $O$. We have
\begin{eqnarray*}
(p+1)\sum_j{L(N_j)^{p+1}} &=& p\sum_j{L(N_j)^{p+1}}+ \sum_j{L(N_j)^{p+1}}\\
&=& p\sum_i{\sum_j{x_{ij}w_{ij}L(N_j)^p}} + \sum_j{L(N_j)^{p+1}}\\
&\leq & \left(\left(\sum_j{L(N_j)^{p+1}}\right)^{\frac{1}{p+1}}+\frac{2p+1}{p+1}\left(\sum_j{L(O_j)^{p+1}}\right)^{\frac{1}{p+1}}\right)^{p+1}.
\end{eqnarray*}
The first equality is obvious, the second one follows by the definition of variables $x_{ij}$ and the inequality follows by inequality (\ref{eq:b-minkowski}).

So, the above inequalities yield
\begin{eqnarray*}
\left(\sum_j{L(N_j)^{p+1}}\right)^{\frac{1}{p+1}} &\leq & \frac{2p+1}{p+1} \frac{1}{(p+1)^{\frac{1}{p+1}}-1} \left(\sum_j{L(O_j)^{p+1}}\right)^{\frac{1}{p+1}}\\
&\leq & \frac{2p+1}{\ln{(p+1)}} \left(\sum_j{L(O_j)^{p+1}}\right)^{\frac{1}{p+1}}\\
&\leq & \frac{2p+1}{\ln{(p+1)}} m^{\frac{1}{p+1}} \max_j{L(O_j)}.
\end{eqnarray*}
The second inequality follows since $e^z\geq z+1$ for $z\geq 0$ and the third one follows by Lemma \ref{lem:lp-norm}.

Now, using Lemma \ref{lem:completion-bcoord} and our last inequality we have
\begin{eqnarray*}
\max_{j,i\in N_j}{{\cal P}(i,N_j)} &\leq & \left(\sum_j{L(N_j)^{p+1}}\right)^{\frac{1}{p+1}} +\max_j{L(O_j)}\\
&\leq & \left(1+\frac{2p+1}{\ln{(p+1)}}m^{\frac{1}{p+1}}\right)\max_j{L(O_j)}.
\end{eqnarray*}
The desired bounds follows by setting $p=\Theta(\log m)$ and $p=1/\epsilon-1$, respectively.\qed
\end{proof}

Note that the game induced by {\sf BCOORD} with $p=1$ is the same with the
game induced by the coordination mechanism {\sf CCOORD} (with $p=1$) that we present in the next section. As such,
it also has a potential function (also similar to the potential function of \cite{FKS05} for linear weighted congestion
games) as we will see in Lemma \ref{lem:psi-potential}. In this way, we obtain a coordination mechanism that induces a potential game, handles anonymous jobs, and has aproximation ratio $O(\sqrt{m})$. Unfortunately,
the next theorem demonstrates that, for higher values of $p$, the Nash dynamics of the game induced
by {\sf BCOORD} may contain a cycle.

\begin{theorem}\label{thm:no-potential}
The game induced by the coordination mechanism {\sf BCOORD} with $p=2$ is not a potential game.
\end{theorem}

Before proving Theorem \ref{thm:no-potential}, we show that the games induced by the coordination mechanisms
{\sf LongestFirst} and {\sf Randomized} may not be potential games either. All the instances presented in the following consist of four machines
and three basic jobs $A$, $B$, and $C$. In each case, we show that the Nash dynamics contain a cycle of moves
of the three basic jobs.

First consider the {\sf LongestFirst} policy and the instance depicted in the following table.

\

\centerline{
\begin{tabular}{|c|c|c|c|}
  \hline
   & A & B & C \\
  $1$ & $14$ & $\infty$ & $5$ \\
  $2$ & $\infty$ & $10$ & $\infty$ \\
  $3$ & $3$ & $9$ & $10$ \\
  $4$ & $7$ & $8$ & $9$ \\
  \hline
\end{tabular}}

\

The cycle is defined on the following states:
\begin{eqnarray*}
& & (C,\underline{B},A,) \rightarrow(C,,\underline{A}B,)\rightarrow(C,,\underline{B},A)\rightarrow(C,,,\underline{A}B)\rightarrow (A\underline{C},,,B)\rightarrow\\
& & (\underline{A},,C,B)\rightarrow(,,A\underline{C},B)\rightarrow(,,A,\underline{B}C)\rightarrow (,B,A,\underline{C})\rightarrow(C,B,A,).
\end{eqnarray*}
Notice that the first and last assignment are the same. In each state, the player that moves next is underlined. Job $B$ is at machine $2$
in the first assignment and has completion time $10$. Hence, it has an incentive to move to machine $3$ (second assignment) where its completion
time is $9$. Job $A$ has completion time $12$ in the second assignment since it is scheduled after job $B$ which has higher load on machine $3$.
Moving to machine $4$ (third assignment), it decreases its completion time to $7$. The remaining moves in the cycle can be verified accordingly.

The instance for the {\sf Randomized} policy contains four additional jobs $D$, $E$, $F$, and $G$ which are always scheduled on
machines $1$, $2$, $3$, and $4$, respectively (i.e., they have infinite load on the other machines). It is depicted in the
following table.

\

\centerline{
\begin{tabular}{|c|c|c|c|c|c|c|c|}
  \hline
   & A & B & C & D & E & F & G\\
  $1$ & $80$ & $\infty$ & $100$ & $2$ & $\infty$ & $\infty$ & $\infty$ \\
  $2$ & $\infty$ & $171$ & $\infty$ & $\infty$ & $2$ & $\infty$ & $\infty$ \\
  $3$ & $2$ & $154$ & $124$ & $\infty$ & $\infty$ & $32$ & $\infty$ \\
  $4$ & $2$ & $76$ & $10$ & $\infty$ & $\infty$ & $\infty$ & 184\\
  \hline
\end{tabular}}

\

The cycle is defined by the same moves of the basic jobs as in the case of {\sf LongestFirst}:
\begin{eqnarray*}
& & (CD,\underline{B}E,AF,G) \rightarrow(CD,E,\underline{A}BF,G)\rightarrow(CD,E,\underline{B}F,AG)\rightarrow(CD,E,F,\underline{A}BG)\rightarrow\\
& & (A\underline{C}D,E,F,BG)\rightarrow(\underline{A}D,E,CF,BG)\rightarrow(D,E,A\underline{C}F,BG)\rightarrow(D,E,AF,\underline{B}CG)\rightarrow\\
& & (D,BE,AF,\underline{C}G)\rightarrow(CD,BE,AF,G).
\end{eqnarray*}
Recall that (see \cite{ILMS05,KP99}) the expected completion time of a job $i$ which is scheduled on machine $j$ in an assignment $N$ is $\frac{1}{2}(w_{ij}+L(N_j))$ when the {\sf Randomized} policy is used.
In each state, the player that moves next is underlined. It can be easily verified that each player in this cycle improves her cost by exactly $1$.
For example, job $B$ has expected completion time $\frac{1}{2}(171+171+2)=172$ at machine $2$ in the first assignment and, hence, an incentive to move to machine $3$
in the second assignment where its completion time is $\frac{1}{2}(154+2+154+32)=171$.

\paragraph{Proof of Theorem \ref{thm:no-potential}.}
Besides the three basic jobs, the instance for the {\sf BCOORD} policy with $p=2$ contains two additional jobs $D$ and $E$ which are always scheduled on
machines $3$ and $4$, respectively. The instance is depicted in the following table.

\

\centerline{
\begin{tabular}{|c|c|c|c|c|c|}
  \hline
   & A & B & C & D & E \\
  $1$ & $4.0202$ & $\infty$ & $4.0741$ & $\infty$ & $\infty$ \\
  $2$ & $\infty$ & $8.2481$ & $\infty$ & $\infty$ & $\infty$ \\
  $3$ & $0.0745$ & $0.6302$ & $0.3078$ & $29.1331$ & $\infty$ \\
  $4$ & $2.4447$ & $5.1781$ & $2.4734$ & $\infty$ & $2.7592$ \\
  \hline
\end{tabular}}

\

The cycle is defined by the same moves of the basic jobs as in the previous cases:
\[(C,\underline{B},AD,E)\rightarrow (C,,\underline{A}BD,E)\rightarrow(C,,\underline{B}D,AE)\rightarrow(C,,D,\underline{A}BE)\rightarrow(A\underline{C},,D,BE)\rightarrow\]
\[(\underline{A},,CD,BE)\rightarrow(,,A\underline{C}D,BE)\rightarrow(,,AD,\underline{B}CE)\rightarrow(,B,AD,\underline{C}E)\rightarrow(C,B,AD,E).\]
Notice that, instead of considering the
completion time $\left(\rho_{ij}\right)^{1/p}L(N_j)$ of a job $i$ on machine $j$ in an assignment $N$, it is equivalent to consider its cost as $w_{ij}L(N_j)^p$.
In this way, we can verify that in any of the moves in the above cycle, the job that moves improves its cost. For example, job $B$
has cost $8.2481^3=561.127758090641$ on machine $2$ in the first assignment and cost $0.6302(0.0745+0.6302+29.1331)^2=561.063473430968$ on machine $3$ in the second assignment.
\qed

\section{The coordination mechanism {\sf CCOORD}}\label{sec:c-cm}
In this section we present and analyze the coordination mechanism {\sf CCOORD} that handles anonymous jobs
and guarantees that the induced game has pure Nash equilibria, price of anarchy at most $O(\log^2 m)$, and price
of stability $O(\log m)$. In order to define the scheduling policy, we first define an interesting family of functions.

\begin{defn}
For integer $k\geq 0$, the function $\Psi_k$ mapping finite sets of reals to the reals is defined as follows:
$\Psi_k(\emptyset)=0$ for any integer $k\geq 1$, $\Psi_0(A)=1$ for any (possibly empty) set $A$, and for any
non-empty set $A=\{a_1, a_2, ..., a_n\}$ and integer $k\geq 1$,
\[\Psi_k(A)=k! \sum_{1\leq d_1 \leq ... \leq d_k \leq n}{\prod_{t=1}^{k}{a_{d_t}}}.\]
\end{defn}
So, $\Psi_k(A)$ is essentially the sum of all possible monomials of total degree $k$ on the elements of $A$. Each
term in the sum has coefficient $k!$. Clearly, $\Psi_1(A)=L(A)$. For $k\geq 2$, compare $\Psi_k(A)$ with $L(A)^k$
which can also be expressed as the sum of the same terms, albeit with different coefficients in $\{1, ..., k!\}$,
given by the multinomial theorem.

The coordination mechanism {\sf CCOORD} schedules job $i$ on machine $j$ in an assignment $N$ so that its completion time
is $${\cal P}(i,N_j) = \left(\rho_{ij}\Psi_p(N_j)\right)^{1/p}.$$

Our proofs extensively use the properties in the next lemma; its proof is given in appendix. The first inequality implies that the schedule defined by {\sf CCOORD}
is always feasible.

\begin{lemma}\label{lem:properties}
For any integer $k\geq 1$, any finite set of non-negative reals $A$, and any non-negative real $b$ the following hold:
\[\begin{array}{l l}
\mbox{a. } L(A)^k \leq \Psi_k(A) \leq k! L(A)^k & \mbox{d. } \Psi_k(A\cup\{b\}) -\Psi_k(A) = k b\Psi_{k-1}(A\cup \{b\})\\
\mbox{b. } \Psi_{k-1}(A)^{k} \leq \Psi_{k}(A)^{k-1} & \mbox{e. } \Psi_k(A) \leq kL(A)\Psi_{k-1}(A)\\
\mbox{c. } \Psi_k(A\cup\{b\}) = \sum_{t=0}^k{\frac{k!}{(k-t)!}b^t\Psi_{k-t}(A)} &\mbox{f. } \Psi_k(A\cup \{b\}) \leq \left(\Psi_k(A)^{1/k}+\Psi_k(\{b\})^{1/k}\right)^k
\end{array}\]
\end{lemma}

The second property implies that $\Psi_k(A)^{1/k} \leq \Psi_{k'}(A)^{1/k'}$ for any integer $k'\geq k$.
The third property suggests an algorithm for computing $\Psi_k(A)$ in time polynomial in $k$ and $|A|$
using dynamic programming.

A careful examination of the definitions of the coordination mechanisms {\sf BCOORD} and {\sf CCOORD} and property (a) in the above lemma, reveals that {\sf CCOORD} makes the completion time of a job assigned to machine $j$ dependent on the approximation $\Psi_p(N_j)^{1/p}$ of the load $L(N_j)$ of the machine instead of its exact load as {\sf BCOORD} does. This will be the crucial tool in order to guarantee that the induced game is a potential game without significantly increasing the price of anarchy. The next lemma defines a potential function on the states of the induced game that will be very useful later.

\begin{lemma}\label{lem:psi-potential}
The function $\Phi(N)=\sum_j{\Psi_{p+1}(N_j)}$ is a potential function for the game induced by the coordination mechanism {\sf CCOORD}.
Hence, this game always has a pure Nash equilibrium.
\end{lemma}

\begin{proof}
Consider two assignments $N$ and $N'$ differing in the strategy of the player controlling job $i$. Assume that job
$i$ is assigned to machine $j_1$ in $N$ and to machine $j_2\not=j_1$ in $N'$. Observe that $N_{j_1}=N'_{j_1}\cup \{w_{ij_1}\}$ and
$N'_{j_2}=N_{j_2}\cup \{w_{ij_2}\}$. By Lemma \ref{lem:properties}d, we have that $\Psi_{p+1}(N_{j_1})-\Psi_{p+1}(N'_{j_1})=(p+1)w_{ij_1}\Psi_p(N_{j_1})$
and $\Psi_{p+1}(N'_{j_2})-\Psi_{p+1}(N_{j_2})=(p+1)w_{ij_2}\Psi_p(N'_{j_2})$. Using these properties and the definitions of the coordination
mechanism {\sf CCOORD} and function $\Phi$,
we have
\begin{eqnarray*}
\Phi(N)-\Phi(N') &=& \sum_j{\Psi_{p+1}(N_j)-\sum_j{\Psi_{p+1}(N'_j)}}\\
&=& \Psi_{p+1}(N_{j_1})+\Psi_{p+1}(N_{j_2})-\Psi_{p+1}(N'_{j_1})-\Psi_{p+1}(N'_{j_2})\\
&=& (p+1)w_{ij_1}\Psi_{p}(N_{j_1})-(p+1)w_{ij_2}\Psi_{p}(N'_{j_2})\\
&=& (p+1)w_{i,\min} \left({\cal P}(i,N_{j_1})^p-{\cal P}(i,N'_{j_2})^p\right)
\end{eqnarray*}
which means that the difference of the potentials of the two assignments and the difference of
the completion time of player $i$ have the same sign as desired.
\qed
\end{proof}

The next lemma relates the maximum completion time of a pure Nash equilibrium to the optimal makespan provided that their potentials are close.
\begin{lemma}\label{lem:completion-time}
Let $O$ be an optimal assignment and let $N$ be a pure Nash equilibrium of the game induced by the coordination mechanism {\sf CCOORD} such that $\left(\Phi(N)\right)^{\frac{1}{p+1}}\leq \gamma\left(\Phi(O)\right)^{\frac{1}{p+1}}$. Then,
\[\max_{j,i\in N_j}{{\cal P}(i,N_j)} \leq \left(\gamma (p+1)m^{\frac{1}{p+1}}+p\right)\max_j{L(O_j)}.\]
\end{lemma}

\begin{proof}
Let $i^*$ be the job that has the maximum completion time in $N$. Denote by $j_1$ the machine $i^*$ uses in assignments $N$ and let $j_2$ be a machine such that $\rho_{i^*j_2}=1$. If $j_1=j_2$, the definition of the coordination mechanism {\sf CCOORD} and Lemma \ref{lem:properties}b yield
\begin{eqnarray}\nonumber
\max_{j, i\in N_j}{{\cal P}(i,N_j)} &=& {\cal P}(i^*,N_{j_1})\\\nonumber
&= & \Psi_p(N_{j_1})^{1/p}\\\nonumber
&\leq & \Psi_{p+1}(N_{j_1})^{\frac{1}{p+1}}\\\label{eq:c-same-machines}
&\leq & \left(\sum_j{\Psi_{p+1}(N_j)}\right)^{\frac{1}{p+1}}.
\end{eqnarray}
Otherwise, since player $i$ has no incentive to use machine $j_2$ instead of $j_1$, we have
\begin{eqnarray}\nonumber
\max_{j, i\in N_j}{{\cal P}(i,N_j)} &=& {\cal P}(i^*,N_{j_1})\\\nonumber
&\leq & {\cal P}(i^*, N_{j_2} \cup \{w_{i^*j_2}\}) \\\nonumber
&=& \Psi_p(N_{j_2}\cup\{w_{i^*j_2}\})^{1/p}\\\nonumber
&\leq & \Psi_p(N_{j_2})^{1/p}+\Psi_p(\{w_{i^*j_2}\})^{1/p}\\\nonumber
&\leq & \Psi_{p+1}(N_{j_2})^{\frac{1}{p+1}}+(p!)^{1/p}w_{i^*j_2}\\\nonumber
&= & \Psi_{p+1}(N_{j_2})^{\frac{1}{p+1}}+(p!)^{1/p}\min_j{w_{i^*j}}\\\label{eq:c-diff-machines}
&\leq & \left(\sum_j{\Psi_{p+1}(N_j)}\right)^{\frac{1}{p+1}}+p\max_j{L(O_j)}.
\end{eqnarray}
The first two equalities follows by the definition of {\sf CCOORD}, the first inequality follows since player $i^*$ has no incentive to use machine $j_2$ instead of $j_1$, the second inequality follows by Lemma \ref{lem:properties}f, the third inequality follows by Lemma \ref{lem:properties}b and the definition of function $\Psi_p$, the third equality follows by the definition of machine $j_2$ and the last inequality is obvious.

Now, observe that the term in parenthesis in the rightmost side of inequalities (\ref{eq:c-same-machines}) and (\ref{eq:c-diff-machines}) equals the potential $\Phi(N)$. Hence, in any case, we have
\begin{eqnarray*}
\max_{j, i\in N_j}{{\cal P}(i,N_j)} &\leq& (\Phi(N))^{\frac{1}{p+1}}+p\max_j{L(O_j)}\\
&\leq & \gamma (\Phi(O))^{\frac{1}{p+1}}+p\max_j{L(O_j)}\\
&=& \gamma\left(\sum_j{\Psi_{p+1}(O_j)}\right)^{\frac{1}{p+1}}+p\max_j{L(O_j)}\\
&\leq &\gamma\left((p+1)! \sum_j{L(O_j)^{p+1}}\right)^{\frac{1}{p+1}}+p\max_j{L(O_j)}\\
&\leq & \left(\gamma (p+1)m^{\frac{1}{p+1}}+p\right)\max_j{L(O_j)}.
\end{eqnarray*}
The second inequality follows by the inequality on the potentials of assignments $N$ and $O$, the equality follows by the definition of the potential function $\Phi$, the third inequality follows by Lemma \ref{lem:properties}a and the last one follows by Lemma \ref{lem:lp-norm}.
\qed
\end{proof}

A first application of Lemma \ref{lem:completion-time} is in bounding the price of stability of the induced game.

\begin{theorem}\label{lem:stability}
The game induced by the coordination mechanism {\sf CCOORD} with $p=\Theta(\log m)$ has price of stability at most $O(\log m)$.
\end{theorem}

\begin{proof}
Consider the optimal assignment $O$ and the pure Nash equilibrium $N$ of minimum potential. We have
$\left(\Phi(N)\right)^{\frac{1}{p+1}}\leq \left(\Phi(O)\right)^{\frac{1}{p+1}}$ and, using Lemma \ref{lem:completion-time},
we obtain that the maximum completion time in $N$ is at most $(p+1)m^{\frac{1}{p+1}}+p$ times the makespan of $O$.
Setting $p=\Theta(\log m)$,
the theorem follows.
\qed
\end{proof}

A second application of Lemma \ref{lem:completion-time} is in bounding the price of anarchy. In order to apply it, we need a relation between the potential of an equilibrium and the potential of an optimal assignment; this is provided by the next lemma.

\begin{lemma}\label{lem:equilibrium}
Let $O$ be an optimal assignment and $N$ be a pure Nash equilibrium of the game induced by the coordination mechanism {\sf CCOORD}. Then,
\[\left(\Phi(N)\right)^{\frac{1}{p+1}} \leq \frac{p+1}{\ln{2}} \left(\Phi(O)\right)^{\frac{1}{p+1}}.\]
\end{lemma}

\begin{proof}
Consider a pure Nash equilibrium $N$ and an optimal assignment $O$. Since no job has an incentive to change her strategy from $N$, for any job $i$ that
is assigned to machine $j_1$ in $N$ and to machine $j_2$ in $O$, we have that
\begin{eqnarray*}\left(\rho_{ij_1}\Psi_p(N_{j_1})\right)^{1/p} &\leq & \left(\rho_{ij_2}\Psi_p(N_{j_2}\cup \{w_{ij_2}\})\right)^{1/p}.
\end{eqnarray*}
Equivalently, by raising both sides to the power $p$ and multiplying both sides with $w_{i,\min}$, we have that
\begin{eqnarray*}
w_{ij_1}\Psi_p(N_{j_1}) &\leq & w_{ij_2}\Psi(N_{j_2}\cup \{w_{ij_2}\}).
\end{eqnarray*}
Using the binary variables $x_{ij}$ and $y_{ij}$ to denote whether
job $i$ is assigned to machine $j$ in the assignment $N$ ($x_{ij}=1$) and $O$ ($y_{ij}=1$) or not ($x_{ij}=0$ and $y_{ij}=0$, respectively), we can express the last inequality
as follows:
\begin{eqnarray*}
\sum_j{x_{ij}w_{ij}\Psi_p(N_{j})} &\leq & \sum_j{y_{ij}w_{ij}\Psi(N_{j}\cup \{w_{ij}\})}
\end{eqnarray*}
By summing over all jobs, we have
\begin{eqnarray*}
\sum_i{\sum_j{x_{ij}w_{ij}\Psi_p(N_{j})} } &\leq & \sum_i{\sum_j{y_{ij}w_{ij}\Psi(N_{j}\cup \{w_{ij}\})}}
\end{eqnarray*}
By exchanging the double sums and since $\sum_i{x_{ij}w_{ij}}=L(N_j)$, we obtain
\begin{eqnarray}\label{ineq:delta}
\sum_j{L(N_j)\Psi_p(N_j)} &\leq & \sum_j{\sum_{i}{y_{ij}w_{ij}\Psi_p(N_j \cup \{w_{ij}\})}}
\end{eqnarray}

We now work with the potential of assignment $N$. We have
\begin{eqnarray*}
2\Phi(N) &= & \Phi(N)+\sum_j{\Psi_{p+1}(N_j)}\\
&\leq & \Phi(N)+(p+1)\sum_j{L(N_j)\Psi_p(N_j)}\\
&\leq & \Phi(N)+ (p+1)\sum_j{\sum_i{y_{ij}w_{ij}}\Psi_p(N_j\cup \{w_{ij}\})}\\
&=& \Phi(N)+(p+1)\sum_j{\sum_i{y_{ij}w_{ij}}\sum_{t=0}^{p}{\frac{p!}{(p-t)!}\Psi_{p-t}(N_j)w_{ij}^t}}\\
&=& \Phi(N)+\sum_j{\sum_{t=0}^{p}{\frac{(p+1)!}{(p-t)!}\Psi_{p-t}(N_j)\sum_i{y_{ij}w_{ij}^{t+1}}}}\\
&\leq & \Phi(N)+\sum_j{\sum_{t=0}^{p}{\frac{(p+1)!}{(p-t)!(t+1)!}\Psi_{p-t}(N_j)\Psi_{t+1}(O_j)}}\\
&=& \Phi(N)+\sum_j{\sum_{t=1}^{p+1}{\left(\begin{array}{c}p+1\\t\end{array}\right)\Psi_{p+1-t}(N_j)\Psi_{t}(O_j)}}\\
&\leq & \Phi(N)+\sum_j{\sum_{t=1}^{p+1}{\left(\begin{array}{c}p+1\\t\end{array}\right)\Psi_{p+1}(N_j)^{\frac{p+1-t}{p+1}}\Psi_{p+1}(O_j)^{\frac{t}{p+1}}}}\\
&=& \Phi(N)+\sum_j{\left(\left(\Psi_{p+1}(N_j)^{\frac{1}{p+1}}+\Psi_{p+1}(O_j)^{\frac{1}{p+1}}\right)^{p+1}-\Psi_{p+1}(N_j)\right)}\\
&=& \Phi(N)+\sum_j{\left(\Psi_{p+1}(N_j)^{\frac{1}{p+1}}+\Psi_{p+1}(O_j)^{\frac{1}{p+1}}\right)^{p+1}}-\sum_j{\Psi_{p+1}(N_j)}\\
&\leq & \left(\left(\sum_j{\Psi_{p+1}(N_j)}\right)^{\frac{1}{p+1}}+\left(\sum_j{\Psi_{p+1}(O_j)}\right)^{\frac{1}{p+1}}\right)^{p+1}\\
&=& \left(\left(\Phi(N)\right)^{\frac{1}{p+1}}+\left(\Phi(O)\right)^{\frac{1}{p+1}}\right)^{p+1}
\end{eqnarray*}
The first inequality follows by Lemma \ref{lem:properties}e, the second inequality follows by inequality (\ref{ineq:delta}), the second equality follows by Lemma \ref{lem:properties}c, the third equality follows by exchanging the sums,
the third inequality follows since the jobs $i$ assigned to machine $j$ are those for which $y_{ij}=1$ and by the definition of function $\Psi_{t+1}$ which yields that
$\Psi_{t+1}(O_j)\geq (t+1)! \sum_i{y_{ij}w_{ij}^{t+1}}$, the fourth equality follows by updating the limits of the sum over $t$, the fourth inequality follows by Lemma \ref{lem:properties}b,
the fifth equality follows by the binomial theorem, the sixth equality is obvious, the fifth inequality follows by Minkowski inequality (Lemma \ref{lem:minkowski}) and by the definition of the potential $\Phi(N)$, and the last equality follows by the definition of the potentials $\Phi(N)$ and $\Phi(O)$.

By the above inequality, we obtain that
\begin{eqnarray*}
\left(\Phi(N)\right)^{\frac{1}{p+1}} &\leq & \frac{1}{2^{\frac{1}{p+1}}-1} \left(\Phi(O)\right)^{\frac{1}{p+1}}\leq \frac{p+1}{\ln{2}} \left(\Phi(O)\right)^{\frac{1}{p+1}}
\end{eqnarray*}
where the last inequality follows using the inequality $e^z\geq z+1$.
\qed
\end{proof}
We are now ready to bound the price of anarchy.
\begin{theorem}\label{thm:c-cm-poa}
The price of anarchy of the game induced by the coordination mechanism {\sf CCOORD} with $p=\Theta(\log m)$ is $O\left(\log^2 m\right)$. Also, for every constant $\epsilon\in (0,1/2]$, the price of anarchy of the game induced by the coordination mechanism {\sf CCOORD} with $p=1/\epsilon-1$ is $O\left(m^\epsilon\right)$.
\end{theorem}

\begin{proof}
Consider a pure Nash equilibrium $N$ and let $O$ be the optimal assignment.
Using Lemma \ref{lem:equilibrium}, we have that $\left(\Phi(N)\right)^{\frac{1}{p+1}} \leq \frac{p+1}{\ln{2}} \left(\Phi(O)\right)^{\frac{1}{p+1}}$.
Hence, by Lemma \ref{lem:completion-time}, we obtain that the maximum completion time in $N$ is at most $\frac{(p+1)^2}{\ln 2}m^{\frac{1}{p+1}}+p$ times the makespan of $O$. By setting $p=\Theta(\log m)$ and $p=1/\epsilon-1$, respectively, the theorem follows.
\qed
\end{proof}

\section{Discussion and open problems}\label{sec:open}
Our focus in the current paper has been on pure Nash equilibria. It is also interesting to generalize the bounds on the price of anarchy of the games induced by our coordination mechanisms for mixed Nash equilibria. Recently, Roughgarden \cite{R09} defined general smoothness arguments that can be used to bound the price of anarchy of games having particular properties. Bounds on the price of anarchy over pure Nash equilibria that are proved using smoothness arguments immediately imply that the same bounds on the price of anarchy hold for mixed Nash equilibria as well. We remark that the arguments used in the current paper in order to prove our upper bounds are not smoothness arguments. At least in the case of the coordination mechanism {\sf BCOORD}, smoothness arguments cannot be used to prove a bound on the price of anarchy as small as $O\left(\frac{\log m}{\log\log m}\right)$ since the price of anarchy over mixed Nash equilibria is provably higher in the case. We demonstrate this using the following construction. Czumaj and Voecking \cite{CV02} present a game induced by the {\sf Makespan} policy on related machines which has price of anarchy over mixed Nash equilibria at least $\Omega\left(\frac{\log m}{\log \log\log m}\right)$. The instance used in \cite{CV02} consists of $n$ jobs and $m$ machines. Each machine $j$ has a speed $\alpha_j\geq 1$ with $\alpha_1=1$ and each job $i$ has a weight $w_i$. The processing time of job $i$ on machine $j$ is $w_{ij}=\alpha_j w_i$ (i.e., the inefficiencies of the jobs are the same on the same machine). Now, consider the game induced by the coordination mechanism {\sf BCOORD} for the instance that consists of the same machines and jobs in which the processing time of job $i$ on machine $j$ is defined by $w'_{ij}=\alpha_j^{\frac{p}{p+1}} w_i$, i.e., the inefficiency of any job on machine $j$ is $\alpha_j^{\frac{p}{p+1}}$. Here, $p$ is the parameter used by {\sf BCOORD}. By the definition of {\sf BCOORD}, we can easily see that the game induced is identical with the game induced by {\sf Makespan} on the original instance of \cite{CV02}. Also note that, in our instance, the processing time of the jobs is not increased (and, hence, the optimal makespan is not larger than that in the original instance of \cite{CV02}). Hence, the lower bound of \cite{CV02} implies a lower bound on the price of anarchy over mixed Nash equilibria of the game induced by the coordination mechanism {\sf BCOORD}.

Our work reveals several other interesting questions. First of all, it leaves open the question of whether coordination mechanisms with constant approximation
ratio exist. In particular, is there any coordination mechanism that handles anonymous jobs,
guarantees that the induced game has pure Nash equilibria, and has constant price of anarchy?
Based on the lower bounds of \cite{AJM08,FS10}, such a coordination mechanism (if it exists) must
use preemption. Alternatively, is the case of anonymous jobs provably more difficult than the case
where jobs have IDs? Investigating the limits of non-preemptive mechanisms is still interesting.
Notice that {\sf AJM-1} is the only non-preemptive coordination mechanism
that has approximation ratio $o(m)$ but it does not guarantee that the induced game has
pure Nash equilibria; furthermore, the only known non-preemptive coordination mechanism that induces a potential
game with bounded price of anarchy is {\sf ShortestFirst}. So, is there any non-preemptive (deterministic or randomized)
coordination mechanism that is simultaneously $o(m)$-approximate and induces a potential game?
We also remark that Theorem \ref{thm:no-potential} does not necessarily exclude a game
induced by the coordination mechanism {\sf BCOORD} from having pure Nash equilibria. Observe that
the examples in the proof of Theorem \ref{thm:no-potential} do not consist of best-response moves and,
hence, it is interesting to investigate whether best-response moves converge to pure Nash equilibria
in such games.

Furthermore, we believe that the games induced by the coordination mechanism
{\sf CCOORD} are of independent interest. We have proved that these games belong to the class
{\sf PLS} \cite{JPY88}. Furthermore, the result of Monderer and Shapley \cite{MS96} and the proof
of Lemma \ref{lem:psi-potential} essentially show that each of these games is isomorphic to a congestion
game. However, they have a beautiful definition as games on parallel machines that gives them a particular
structure. What is the complexity of computing pure Nash equilibria in such games?
Even in case that these games are {\sf PLS}-complete (informally, this would mean that computing a pure Nash
equilibrium is as hard as finding any object whose existence is guaranteed by a potential function
argument) like several variations of congestion games that were considered recently \cite{ARV06,FPT04,SV08},
it is still interesting to study the convergence time to efficient assignments.
A series of recent papers \cite{AAEMS08,CMS06,FFM08,MV04} consider adversarial
rounds of best-response moves in potential games so that each player is given at least one chance to
play in each round (this is essentially our assumption in Lemma \ref{lem:a-cm-potential-convergence}
for the coordination mechanism {\sf ACOORD}). Does the game induced by the coordination mechanism
{\sf CCOORD} converges to efficient assignments after a polynomial number of adversarial rounds of best-response moves?
Although it is a potential game, it does not have the particular properties considered in \cite{AAEMS08}
and, hence, proving such a statement probably requires different techniques.

Finally, recall that we have considered the maximum completion time as the measure of the efficiency of schedules. Other measures such as the weighted sum of completion times that is recently studied in \cite{CGM10} are interesting as well. Of course, considering the application of coordination mechanisms to settings different than scheduling is an important research direction.

\small\paragraph{Acknowledgments.} I would like to thank Chien-Chung Huang for helpful comments on an earlier draft of the paper.


\appendix
\newpage\normalsize\section{Proof of Lemma \ref{lem:properties}}
The properties clearly hold if $A$ is empty or $k=1$. In the following, we assume that $k\geq 2$ and $A=\{a_1, ..., a_n\}$ for integer $n\geq 1$.
\paragraph{a.} Clearly,
\begin{eqnarray*}
L(A)^k &=& \left(\sum_{t=1}^k{a_t}\right)^k = \sum_{1\leq d_1 \leq ... \leq d_k \leq n}{\zeta(d_1, ..., d_k)\prod_{t=1}^{k}{a_{d_t}}}
\end{eqnarray*}
where $\zeta(d_1, ..., d_k)$ are multinomial coefficients on $k$ and, hence, belong $\{1, ..., k!\}$. The property then follows by the definition of $\Psi_k(A)$.

\paragraph{b.} We can express $\Psi_{k-1}(A)^k$ and $\Psi_k(A)^{k-1}$ as follows:
\begin{eqnarray*}
\Psi_{k-1}(A)^k &=& \left((k-1)!\right)^k \left(\sum_{1\leq d_1 \leq ... \leq d_k \leq n}{\prod_{t=1}^{k}{a_{d_t}}}\right)^k \\
&= & \left((k-1)!\right)^k \sum_{1\leq d_1 \leq ... \leq d_{k(k-1)} \leq n}{\zeta_1(d_1, ..., d_{k(k-1)})\prod_{t=1}^{k(k-1)}{a_{d_t}}}.\\
\Psi_k(A)^{k-1} &=& \left(k!\right)^{k-1} \left(\sum_{1\leq d_1 \leq ... \leq d_k \leq n}{\prod_{t=1}^{k}{a_{d_t}}}\right)^{k-1}\\
&= & \left(k!\right)^{k-1} \sum_{1\leq d_1 \leq ... \leq d_{k(k-1)} \leq n}{\zeta_2(d_1, ..., d_{k(k-1)})\prod_{t=1}^{k(k-1)}{a_{d_t}.}}
\end{eqnarray*}
So, both $\Psi_{k-1}(A)^k$ and $\Psi_k(A)^{k-1}$ are sums of all monomials of degree $k(k-1)$ over the elements of $A$ with different coefficients.
The coefficient $\zeta_1(d_1, ..., d_{k(k-1)})$ is the number of different ways to partition the multiset $D=\{d_1, ..., d_{k(k-1)}\}$ of size $k(k-1)$
into $k$ disjoint ordered multisets each of size $k-1$ so that the union of the ordered multisets yields the original multiset. We refer to
these partitions as $(k,k-1)$-partitions. The
coefficient $\zeta_2(d_1, ..., d_{k(k-1)})$ is the number of different ways to partition $D$ into $k-1$ disjoint ordered multisets
each of size $k$ (resp. $(k-1,k)$-partitions). Hence, in order to prove the property, it suffices to show that for any multiset $\{d_1, ..., d_{k(k-1)}\}$,
\begin{eqnarray}\label{ineq:partitions}
\frac{\zeta_1(d_1, ..., d_{k(k-1)})}{\zeta_2(d_1, ..., d_{k(k-1)})} \leq \frac{k^{k-1}}{(k-1)!}.
\end{eqnarray}
Assume that some element of $D$ has multiplicity more than one and consider the new multiset $D'=\{d_1, ..., d'_i, ..., d_{k(k-1)}\}$ that
replaces one appearance $d_i$ of this element with a new element $d'_i$ different than all elements in $D$. Then, in order to generate all
$(k,k-1)$-partitions of $D'$, it suffices to consider the $(k,k-1)$-partitions of $D$ and, for each of them, replace $d_i$ with $d'_i$ once for each
of the ordered sets in which $d_i$ appears. Similarly, we can generate all $(k-1,k)$-partitions of $D'$ using the $(k-1,k)$-partitions of $D$.
Since the number of sets in $(k,k-1)$-partitions is larger than the number of sets in $(k-1,k)$-partitions, we will have that
\begin{eqnarray*}
\frac{\zeta_1(d_1, ..., d'_i, ..., d_{k(k-1)})}{\zeta_2(d_1, ..., d'_i, ..., d_{k(k-1)})} \geq \frac{\zeta_1(d_1, ..., d_{k(k-1)})}{\zeta_2(d_1, ..., d_{k(k-1)})} .
\end{eqnarray*}

By repeating this argument, we obtain that the ratio at the left-hand side of inequality (\ref{ineq:partitions})
is maximized when all $d_i$'s are distinct. In this case, both $\zeta_1$ and $\zeta_2$ are given by the multinomial coefficients
\[\zeta_1(d_1, ..., d_{k(k-1)}) = \left(\begin{array}{c}k(k-1)\\\underbrace{k-1, ..., k-1}_{\mbox{\tiny $k$ times}}\end{array}\right) = \frac{(k(k-1))!}{((k-1)!)^{k}}\]
and
\[\zeta_1(d_1, ..., d_{k(k-1)}) = \left(\begin{array}{c}k(k-1)\\\underbrace{k, ..., k}_{\tiny \mbox{$k-1$ times}}\end{array}\right) = \frac{(k(k-1))!}{(k!)^{k-1}}\]
and their ratio is exactly the one at the right-hand side of the inequality (\ref{ineq:partitions}).

\paragraph{c.} The property follows easily by the definition of function $\Psi_k$ by observing that all the monomials of degree $k$ over
the elements of $A$ that contain $b^t$ are generated by multiplying $b^t$ with the terms of $\Psi_{k-t}(A)$.

\paragraph{d.} By property (c), we have
\begin{eqnarray*}
\Psi_k(A\cup\{b\}) -\Psi_k(A) &=& \sum_{t=1}^k{\frac{k!}{(k-t)!}b^t\Psi_{k-t}(A)}\\
&=& kb\sum_{t=1}^{k}{\frac{(k-1)!}{(k-t)!}b^{t-1}\Psi_{k-t}(A)}\\
&=& kb\sum_{t=0}^{k-1}{\frac{(k-1)!}{(k-1-t)!}b^{t}\Psi_{k-1-t}(A)}\\
&=& kb\Psi_{k-1}(A\cup\{b\}).
\end{eqnarray*}

\paragraph{e.} Working with the right-hand side of the inequality and using the definitions of $L$ and $\Psi_{k-1}$, we have
\begin{eqnarray*}
k L(A)\Psi_{k-1}(A) &=& k!\left(\sum_{t=1}^n{a_t}\right)\cdot \sum_{1\leq d_1 \leq ... \leq d_{k-1}\leq n}{\prod_{t=1}^{k-1}{a_{d_t}}}\\
&\geq & k!\sum_{1\leq d_1 \leq ... \leq d_{k}\leq n}{\prod_{t=1}^{k}{a_{d_t}}}\\
&=& \Psi_{k}(A).
\end{eqnarray*}
The equalities follow obviously by the definitions. To see why the inequality holds, observe that the multiplication of the sum of all monomials of degree $1$ with the sum of all monomials of degree $k-1$
will be a sum of all monomials of degree $k$, each having coefficient at least $1$. \qed

\paragraph{f.} The proof follows by the derivation below in which we use property (c), the fact that $\Psi_t(\{b\})=t! b^t$ by the definition of function $\Psi_t$, property (b), and the binomial theorem. We have
\begin{eqnarray*}
\Psi_k(A\cup\{b\}) &=& \sum_{t=0}^k{\frac{k!}{(k-t)!}b^t\Psi_{k-t}(A)}\\
&=& \sum_{t=0}^k{\left(\begin{array}{c}k\\t\end{array}\right)\Psi_{k-t}(A)\Psi_t(\{b\})}\\
&\leq & \sum_{t=0}^k{\left(\begin{array}{c}k\\t\end{array}\right)\Psi_{k}(A)^\frac{k-t}{k}\Psi_k(\{b\})^\frac{t}{k}}\\
&=& \left(\Psi_k(A)^{1/k}+\Psi_k(\{b\})^{1/k}\right)^k.
\end{eqnarray*}
\qed\end{document}